\documentclass[runningheads]{llncs}
\usepackage[T1]{fontenc}

\usepackage{graphicx}

\begin{document}

\title{String Rearrangement Inequalities
  and \\a Total Order Between Primitive Words
  \thanks{Corresponding author: Kai Jin~(\email{cscjjk@gmail.com}).\quad Supported by National Natural Science Foundation of China 62002394.}
}

\titlerunning{Rearrangement Inequalities for String Concatenation}

\author{Ruixi Luo\inst{1}\orcidID{0000-0003-0483-0119} \and \\
Taikun Zhu\inst{1}\orcidID{0000-0001-7365-9576} \and \\
Kai Jin\inst{1}\orcidID{0000-0003-3720-5117}}

\authorrunning{R. Luo et al.} 

\institute{School of Intelligent Systems Engineering, Sun Yat-Sen University, Shenzhen, China
\email{{luorx,zhutk3}@mail2.sysu.edu.cn},\email{cscjjk@gmail.com}}

\maketitle
\begin{abstract} 
We study the following rearrangement problem: Given $n$ words, rearrange and concatenate them so that the obtained string is lexicographically smallest (or largest, respectively). We show that this problem reduces to sorting the given words so that their repeating strings are non-decreasing (or non-increasing, respectively), where the repeating string of a word $A$ refers to the infinite string $AAA\ldots$. Moreover, for fixed size alphabet $\Sigma$, we design an $O(L)$ time sorting algorithm of the words (in the mentioned orders), where $L$ denotes the total length of the input words. Hence we obtain an $O(L)$ time algorithm for the rearrangement problem. Finally, we point out that comparing primitive words via comparing their repeating strings leads to a total order, which can further be extended to a total order on the finite words (or all words).

\keywords{String rearrangement inequalities \and Primitive words \and Combinatorics on words \and String ordering \and Greedy algorithm.}
\end{abstract}

\section{Introduction}

Combinatorics on words (MSC: 68R15) have strong connections to many fields of mathematics and have found significant applications
  to theoretical computer science and molecular biology (DNA sequences) \cite{KMP,book-bible,pattern,Run-2,Run-3,LW-1}.
Particularly, the primitive words over some alphabet $\Sigma$ have received special interest,
   as they have applications in the formal languages and algebraic theory of codes \cite{primitive-contextfree-96,primitive-contextfree-91,primitive-contextfree-14,primitive-roots}.
A word is \emph{primitive} if it is not a proper power of a shorter word.

In this paper, we consider the following rearrangement problem of words:
Given $n$ words $A_1,\ldots,A_n$, rearrange and concatenate these words so that
  the obtained string $S$ is lexicographically smallest (or largest, respectively).
We prove that the lexicographical smallest outcome of $S$ happens
   when the words are arranged so that their repeating strings are increasing,
and the largest outcome of $S$ happens when the words are arranged reversely; see Lemma~\ref{lemma:string-rearrangement-inequalities}.
Throughout, the \emph{repeating string} of a word $A$ refers to the infinite string $R(A)=AAA\ldots$.

Based on the above lemma (we suggest to name its results as ``string rearrangement inequalities''),
  the aforementioned rearrangement problem reduces to sorting the words $A_1,\ldots,A_n$ so that
   $R(A_1)\leq \ldots \leq R(A_n)$.
We show how to sort for the special case where $A_1,\ldots,A_n$ are primitive and distinct in $O(\sum_i |A_i|)$ time.
The general case can be easily reduced to the special case and can be solved in the same time bound.
Note that we assume bounded alphabet $\Sigma$ and the size of $\Sigma$ is fixed.
Moreover, $|X|$ always denotes the length of word $X$.

Our algorithm beats the plain algorithm based on sorting (via comparing several pairs $R(A_i),R(A_j)$) by a factor of $\log n$.
The algorithm is simple -- it only applies basic data structures such as tries and the failure function \cite{KMP}. Nevertheless, its  correctness and running time analysis is non-straightforward.

\newcommand{\Oinf}{\leq_\infty}

\medskip We mention that comparing primitive words via comparing their repeating strings leads to a total order $\Oinf$ on primitive words, which can extended to a total order $\Oinf$ on all words (section~\ref{sect:order}).
We show that this order is the same as the lexicographical order over Lyndon words
  but are different over primitive words and finite words.
It is also different from \emph{reflected lexicographic order}, \emph{co-lexicographic order}, \emph{shortlex order},
  \emph{Kleene-Brouwer order}, \emph{V-Order}, \emph{alternative order} \cite{order-1,order-2,order-3,order-4,order-5}.
It seems that order $\Oinf$ has not been reported in literature.
		
\subsection{Related work}

It is shown in \cite{primitive-contextfree-97} that
  the language of Lyndon words is not context-free.
Also, many people conjectured that the language of primitive words is not context-free \cite{primitive-contextfree-91,primitive-contextfree-96,primitive-contextfree-97,primitive-contextfree-14}.
But this conjecture is unsettled thus far, to the best of our knowledge.
It would be interesting to explore whether the results shown in this paper can be helpful for solving this longstanding open problem in the future. See more introductions about primitive words in \cite{primitive-roots}.

\medskip Fredricksen and Maiorana \cite{LW-4-debrujin-original} showed that
  if one concatenates, in lexicographic order, all the Lyndon words that have length dividing a given number $n$,
    the result is a de Bruijn sequence.
    Au \cite{LW-3-debrujin-primitive} further showed that if ``dividing $n$'' is replaced by ``identical to $n$'',
     the result is a sequence which contains exactly once every primitive word of length $n$ as a factor.
Note that concatenating some Lyndon words by lexicographic order is the same as concatenating by $\Oinf$ order.

\medskip The Lyndon words have many interesting properties and have found plentiful applications, both theoretically and practically.
Among others, they are used in constructing de Brujin sequence as mentioned above (which have found applications in cryptography), and they are applied in proving the ``runs theorem'' \cite{Run-1,Run-2,Run-3}.
The famous Chen-Fox-Lyndon Theorem states that any word $W$ can be uniquely factorized into $W = W_1 W_2 \ldots W_m$,
  such that each $W_i$ is a Lyndon word, and $W_1 \geq \ldots \geq W_m$ \cite{LW-1,LW-2} (Here $\geq$ refers to the opposite of Lexicographical order, but is the same as the opposite of $\Oinf$).
This factorization is used in the computation of runs in a word \cite{Run-3}.
See the Bible of combinatorics on words \cite{book-bible} for more introductions about Lyndon words and primitive words.


\section{Preliminaries}\label{sect:pre}

\begin{definition}\label{def:primitive}
\emph{The $n^{th}$ \emph{\bf power} of word $A$ is defined as:}
$$
	A^n = \left\{\begin{array}{ccc}
		AA^{n-1},& n>0; \\
		\mbox{\emph{empty word}}, & n=0.\\
	\end{array}\right.
$$
A word $A$ is \emph{\bf non-primitive} if it equals $B^k$ for some word $B$ and integer $k\geq 2$.
    Otherwise, $A$ is \emph{\bf primitive}. (By this definition the empty word is not primitive.)
\end{definition}

The next lemma summarizes three results about the powers proved by Lyndon and Sch\"{u}zenberger \cite{equation-1};
 see their Lemmas~3 and 4, and Corollary~4.1.
(More introductions of these results can be found in Section 1.3 ``Conjugacy'' of \cite{book-bible}.)

\begin{lemma}\label{lemma:key-power-lemma}\cite{equation-1}
Given words $A$ and $B$, there exist $C,k,l$ such that $A=C^k$ and $B=C^l$ when one of the following conditions holds:
\begin{enumerate}
\item $AB=BA$.
\item Two powers $A^{m_1}$ and $B^{m_2}$ have a common prefix of length $|A|+|B|$.
\item $A^{m_1}=B^{m_2}$.
\end{enumerate}
\end{lemma}

\newcommand{\rt}{\mathsf{root}}

\begin{definition}\label{def:root}
The \emph{\bf root} of a word $A$, denoted by $\rt(A)$, is the unique primitive word $B$ such that $A$ is a power of $B$.
The uniqueness of root is obvious, a formal proof can be found in Corollary 4.2 of \cite{equation-1} or in \cite{primitive-roots}.
\end{definition}

\begin{lemma}\label{lemma:root}
Assume $A$ is a non-empty word.
Find the largest $j<|A|$ such that the prefix of $A$ with length $j$ equals the suffix of $A$ with length $j$.
Let $k=|A|-j>0$. Then,
\begin{equation}
|\rt(A)| = \left\{
               \begin{array}{clc}
                 k, & ~ &\hbox{$|A| = 0\pmod k$;} \\
                 |A|, & ~ & \hbox{$|A|\neq 0\pmod k$.}
               \end{array}
             \right.
\end{equation}
\end{lemma}

\begin{proof} This result should be well-known. A simple proof is as follows.

\smallskip \noindent \emph{Fact~1}. If $S=BB'=B'B$ and $B,B'$ are non-empty, $S$ is non-primitive.

This is a trivial fact and is implied by Lemma~\ref{lemma:key-power-lemma}~(condition~1); proof omitted.

\smallskip \noindent \emph{Claim~1}. $|\rt(A)|\geq k$.

Proof: The prefix and suffix of $A$ with length $|A|-|\rt(A)|$ are the same, which
  implies that $j\geq |A|-|\rt(A)|$. Consequently, $|\rt(A)|\geq |A|-j=k$.

\smallskip \noindent \emph{Claim~2}. If $|\rt(A)|<|A|$ (i.e., $A$ is non-primitive), then $k\geq |\rt(A)|$.

Proof: Denote $S=\rt(A)$ and assume $|S|<|A|$. Therefore, $A=S^d~(d\geq 2)$.
Suppose to the opposite that $k<|\rt(A)|$.
Let $B$ be the prefix of $S$ with length $k$, and $B'$ be the suffix of $S$ such that $S=BB'$.
As $k<|S|$, we have $j>|A|-|S|\geq |S|$.
Further since the suffix of $A$ with length $j$ (which starts with $B'B$) equals to the prefix of $A$ with length $j$ (which starts with $S=BB'$), we get $S=B'B$. Applying Fact~1, $\rt(A)=S$ is non-primitive. Contradictory.

We are ready to prove the lemma. When $|A|$ is a multiple of $k$,
  $A$ is a power of its prefix of length $k$, which means $|\rt(A)|\leq k$. Further by Claim~1, $|\rt(A)|=k$.
Next, assume $|A|$ is not a multiple of $k$.
Since $|A|$ is a multiple of $|\rt(A)|$, we see $|\rt(A)|\neq k$. Further by Claims~1 and 2, it follows that $|\rt(A)|=|A|$.
\qed
\end{proof}

For a non-empty word $A$, denote by $R(A)$ the infinite repeating string $AA\ldots$.

\begin{quote}
\textbf{Problem~1.}
Given non-empty words $A_1,\ldots,A_n$, sort them so that
    $$R(A_1)\leq \ldots \leq R(A_n).$$
\end{quote}

Clearly, $R(A)=R(\rt(A))$. To solve Problem~1, we can replace $A$ by $\rt(A)$ (using a preprocessing algorithm based on Lemma~\ref{lemma:root}), and then it reduces to:

\begin{quote}
\textbf{Problem~1'.}
Given primitive words $A_1,\ldots,A_n$, sort them so that
    $$R(A_1)\leq \ldots \leq R(A_n).$$
\end{quote}

\newcommand{\SA}{\mathcal{A}}

\begin{definition}\label{def:deg}
For any two non-empty words $S$ and $A$,
 denote by $\deg_A(S)$ the largest integer $d$ so that $S^d$ is a prefix of $A$.
 Moreover, for non-empty word $S$ and set of non-empty words $\SA=\{A_1,\ldots,A_n\}$, denote $\deg_\SA(S)= \max_j \deg_{A_j}(S)$.

In other words, if we build the trie $T$ of $\SA$, $S^{\deg_\SA(S)}$ is the longest power of $S$ that equals to some path of the trie $T$ starting from its root.

For any $i~(1\leq i\leq n)$, denote
\begin{eqnarray}
N_i &= &\hbox{the $\deg_\SA(A_i)$-th power of $A_i$}\\
M_i = N_iA_i^2 &=& \hbox{the ($\deg_\SA(A_i)$+2)-th power of $A_i$}
\end{eqnarray}

\end{definition}

The following lemma is fundamental to our algorithm.

\begin{lemma}\label{lemma:R(A)R(B)-AB-BA}
For non-empty words $A$ and $B$,
the relation between $R(A)$ and $R(B)$ is the same as the relation between $AB$ and $BA$. In other words,
 \begin{eqnarray}
    R(A)< R(B) & \quad \Leftrightarrow & \quad AB< BA,\\
    R(A)> R(B) & \quad \Leftrightarrow & \quad AB> BA.\label{eqn:>}\\
    R(A)= R(B) & \quad \Leftrightarrow & \quad AB= BA\label{eqn:=},
 \end{eqnarray}
\end{lemma}

\begin{proof}
Assume that $A,B$ are words that consist of the decimal symbols `0',\ldots,`9'.
The proof can be easily extended to the more general case.

Let $\alpha,\beta$ denote the number represented by strings $A,B$. For example, string `89' represents number 89.
Denote $a=|A|$ and $b=|B|$. Observe that
$$ AB<BA ~\Leftrightarrow~ \alpha \cdot 10^b + \beta < \beta \cdot 10^a + \alpha
 ~\Leftrightarrow~  \frac{\alpha}{10^a-1} < \frac{\beta}{10^b-1}.$$ Moreover,
\[  \frac{\alpha}{10^a-1}=\alpha\frac{\frac{1}{10^a}}{1-\frac{1}{10^a}}=\alpha [ \frac{1}{10^a} + (\frac{1}{10^a})^2+ (\frac{1}{10^a})^3 + \ldots ]=0.\alpha \alpha \alpha \cdots=0.\dot{\alpha };\]
\[  \frac{\beta}{10^b-1}=\beta\frac{\frac{1}{10^b}}{1-\frac{1}{10^b}}=\beta [ \frac{1}{10^b} + (\frac{1}{10^b})^2+ (\frac{1}{10^b})^3 + \ldots ]=0.\beta\beta\beta\cdots=0.\dot{\beta}.\]

So, $AB<BA \Leftrightarrow 0.\dot{\alpha } < 0.\dot{\beta} \Leftrightarrow R(A) < R(B).$
Similarly, (\ref{eqn:>}) and (\ref{eqn:=}) hold. \qed
\end{proof}

A more rigorous but complicated proof of Lemma~\ref{lemma:R(A)R(B)-AB-BA} is given in the appendix.

As an interesting corollary of Lemma~\ref{lemma:R(A)R(B)-AB-BA},
  we obtain that ``if $AB\leq BA$ and $BC\leq CB$, then $AC\leq CA$''.
    This transitivity is not obvious without Lemma~\ref{lemma:R(A)R(B)-AB-BA}.

\clearpage
\section{A linear time algorithm for sorting the repeating words}

Assume that $A_1,\ldots,A_n$ are \textbf{primitive}. Denote $L=\sum_i|A_i|$ for short.
This section presents an $O(L)$ time algorithm for solving Problem 1', that is, sorting $R(A_1),\ldots,R(A_n)$.
We start with two nontrivial observations.

\begin{lemma}\label{lemma:relation-M}
The relation between infinitely repeating strings $R(A_i)$ and $R(A_j)$ is the same as the relation between words $M_i$ and $M_j$,
 that is,
 \begin{eqnarray*}
    R(A_i)= R(A_j) & \quad \Leftrightarrow & \quad M_i= M_j,\\
    R(A_i)< R(A_j) & \quad \Leftrightarrow & \quad M_i< M_j,\\
    R(A_i)> R(A_j) & \quad \Leftrightarrow & \quad M_i> M_j.
 \end{eqnarray*}
As a corollary, sorting $R(A_1),\ldots,R(A_n)$ reduces to sorting $M_1,\ldots,M_n$.
\end{lemma}

\begin{proof}
Consider the comparison of $R(A_i)$ and $R(A_j)$.
Assume $|A_i|\leq |A_j|$. Otherwise it is symmetric.

First, consider the case $R(A_i)=R(A_j)$.
Let $m_1=|A_j|$ and $m_2=|A_i|$.
We know $A_i^{m_1}=A_j^{m_2}$ because $R(A_i)=R(A_j)$.
Applying Lemma~\ref{lemma:key-power-lemma} (condition~3), $A_i=C^k$ and $A_j=C^l$ for some $C,k,l$.
Further since $A_i,A_j$ are primitive, $A_i=C=A_j$. It follows that $M_i=M_j$.
Next, assume that $R(A_i)\neq R(A_j)$.

\medskip Let $p=\deg_{A_j}(A_i)$. Thus, $A_j=A_i^pS$, where $p\geq 0$ and $A_i$ is not a prefix of $S$.
Be aware that $p\leq \deg_{\SA}(A_i)$ by the definition of $\deg_{\SA}(A_i)$.

According to Lemma~\ref{lemma:R(A)R(B)-AB-BA}, the comparison of $R(A_i)$ and $R(A_j)$
  equals to the comparison of $A_iA_j$ and $A_jA_i$.
  Further since $A_j=A_i^pS$,
    it equals to the comparison of $A_i^p A_iS$ and $A_i^p SA_i$.
    In the following, we discuss two subcases.

\medskip \noindent \emph{Subcase 1.} \emph{$|S|>|A_i|$, or $|S|\leq |A_i|$ and $S$ is not a prefix of $A_i$.}

Recall that $A_i$ is not a prefix of $S$.
In this subcase, we will find an unequal letter if we compare $A_i$ with $S$ (starting from the leftmost letter).
Comparing $A_i^p A_iS$ and $A_i^p SA_i$ is thus equivalent to comparing the prefixes $A_i^p A_i$ and $A_i^p S$.

Notice that $A_i^p A_i=A_i^{p+1}$ and $A_i^p S=A_j$ are also prefixes of $M_i$ and $M_j$, respectively
(note that $A_i^{p+1}$ is a prefix of $M_i$ because $M_i$ is the $(\deg_{\SA}(A_i)+2)$-th power of $A_i$ and $\deg_{\SA}(A_i)\geq p$ as mentioned above).
Therefore, comparing $M_i$ and $M_j$ is also equivalent to comparing the two prefixes $A_i^p A_i$ and $A_i^p S$.

Altogether, comparing $R(A_i),R(A_j)$ is equivalent to comparing $M_i,M_j$.

\medskip \noindent \emph{Subcase 2.} \emph{$S$ is a prefix of $A_i$. (This means $S$ is a proper prefix of $A_i$ as $S\neq A_i$.)}

Assume $A_i=ST$. Comparing $A_i^p A_iS$ and $A_i^p SA_i$ is just the same as comparing $A_i^p STS$ and $A_i^p SST$.
It reduces to proving that comparing $M_i$ and $M_j$ also reduces to comparing $A_i^p STS$ and $A_i^p SST$.

First, we argue that $ST\neq TS$. Suppose to the opposite that $ST=TS$.
Applying Lemma~\ref{lemma:key-power-lemma} (condition~1), $S=C^k$ and $T=C^l$ for some $C,k,l$.
This implies that $A_i$ and $A_j$ are both powers of $C$, and hence $R(A_i)=R(A_j)$, contradictory.

Observe that $A_i^p STS$ is a prefix of $A_i^p STST=A_i^{p+2}$, which is a prefix of $M_i$
  (because $M_i$ is the $(\deg_\SA(A_i)+2)$-th power of $A_i$ and $\deg_\SA(A_i)+2\geq p+2$).

Observe that $p>0$. Otherwise $A_j=A_i^0 S$ is shorter than $A_i$, which contradicts our assumption $|A_i|\leq |A_j|$.
As a corollary, $A_i$ is a prefix of $A_j=A_i^p S$.
Therefore, $A_i^p SST=A_i^p S A_i= A_jA_i$ is a prefix of $A_j^2$, which is a prefix of $M_j$.

To sum up, $M_i$ and $M_j$ admit $A_i^p STS$ and $A_i^p SST$ as prefixes, respectively.
Further since $TS\neq ST$, comparing $M_i,M_j$ reduces to comparing $A_i^p STS$ and $A_i^p SST$, which
is equivalent to comparing $R(A_i),R(A_j)$ as mentioned above. \qed
\end{proof}

Assume $A_1,\ldots,A_n$ are \textbf{distinct} henceforth in this section. To this end,
 we can use a trie to reduce those duplicate elements in $A_1,\ldots,A_n$, which is trivial.

\begin{lemma}\label{lemma:sum-to-L}
When $A_1,\ldots,A_n$ are primitive and distinct, $\sum_i |N_i| =  O(L)$.
\end{lemma}
\begin{proof}
First, we argue that $N_1,\ldots,N_n$ are distinct.
Suppose that $N_i=N_j~(i\neq j)$.
Recall that $N_i=A_i^m$ (for $m=\deg_\SA(A_i)$) and $N_j=A_j^n$ (for $n=\deg_\SA(A_j)$).
Applying Lemma~\ref{lemma:key-power-lemma} (condition~3), $A_i=C^k$ and $A_j=C^l$.
Further since $A_i,A_j$ are primitive, $A_i=C=A_j$, which contradicts the assumption that $A_i\neq A_j$.

We say $N_i$ \emph{extremal} if it is \textbf{not} a prefix of any word in $\{N_1,\ldots,N_n\}\setminus\{N_i\}$.
Partition $N_1,\ldots,N_n$ into several groups such that
  (a) for elements in the same group, one of them is the prefix of the other, and
  (b) the longest element in each group is extremal.
  (It is obvious that such a partition exists: we can first distribute the extremal ones to different groups,
    and then distribute the non-extremal ones to suitable group (each non-extremal one is a prefix of some extremal ones).

Now, consider any such group, e.g., $N_{i_1},\ldots,N_{i_x}$.
It suffices to prove that (X) $|N_{i_1}|+\ldots + |N_{i_x}| = O(|A_{i_1}|+\ldots+|A_{i_x}|)$, and we prove it in the following.
Without loss of generality, assume that $N_{i_j}$ is a prefix of $N_{i_{j+1}}$ for $j<x$.

We state two important formulas: (i) $N_{i_x}=A_{i_x}$. (ii) $|N_{i_j}|< |A_{i_j}|+|A_{i_{j+1}}|$ for $j<x$.
Equation~(X) above follows from formulas~(i) and (ii) immediately.

\medskip \noindent \emph{Proof of (i).} Suppose to the contrary that $N_{i_x}\neq A_{i_x}$. 
        By the definition of $N_{i_x}$,
           there exists some $A_j$ such that $N_{i_x}$ is a prefix of $A_j$.
           Clearly, $j\neq i_x$ since $N_{i_x}$ is not a prefix of $A_{i_x}$.
           Consequently, $N_{i_x}$ is a prefix of some other $N_j$, which means $N_{i_x}$ is not extremal,
           contradicting property (b) of the grouping mentioned above.

\medskip \noindent \emph{Proof of (ii).} Suppose to the contrary that $|N_{i_j}|\geq |A_{i_j}|+|A_{i_{j+1}}|$.
        Because $N_{i_j}$ and $N_{i_{j+1}}$ are powers of $A_{i,j}$ and $A_{i_{j+1}}$
           and share a common prefix, $N_{i_j}$, of length at least $|A_{i_j}|+|A_{i_{j+1}}|$.
           By Lemma~\ref{lemma:key-power-lemma}~(condition~2),
             $A_{i_j}=C^k$ and $A_{i_{j+1}}=C^l$ for some $C,k,l$. Hence $A_{i_j}=A_{i_{j+1}}$, as $A_{i_j}$ and $A_{i_{j+1}}$ are primitive. Contradictory.

\qed
\end{proof}

Our algorithm for sorting $R(A_1),\ldots,R(A_n)$ is simply as follows.

First, we build a trie of $A_1,\ldots,A_n$ and use it to compute $N_1,\ldots,N_n$.
In particularly, for computing $N_i$, we walk along the trie from the root and search for maximal pieces of $A_i$, which
  takes $O(|N_i|+|A_i|)=O(|N_i|)$ time.
The total running time for computing $N_1,\ldots,N_n$ is therefore $O(\sum_i|N_i|)=O(L)$.

Second, we compute $M_1,\ldots,M_n$ and build a trie of them. By utilizing this trie, we obtain the lexicographic order of $M_1,\ldots,M_n$, which equals the order of $R(A_1),\ldots,R(A_n)$ according to Lemma~\ref{lemma:relation-M}.
The running time of the second step is $\sum_i|M_i|=\sum_i|N_i| + 2\sum_i|A_i|=O(L)+O(L)=O(L)$.

To sum up, we obtain

\begin{theorem}
Problem~1' can be solved in $O(L)=O(\sum_i |A_i|)$ time.
\end{theorem}

In addition, we can solve Problem~1 within the same time bound.
\begin{theorem}\label{thm:sorting-in-linear}
Problem~1 can be solved in $O(L)=O(\sum_i |A_i|)$ time.
\end{theorem}
\begin{proof}
It remains to showing that $\rt(A_i)$ can be computed in $O(|A_i|)$ time.

Applying Lemma~\ref{lemma:root}, computing $\rt(A)$ reduces to
  finding the largest $j<|A|$ such that the prefix of $A$ with length $j$ equals the suffix of $A$ with length $j$.
Moreover, the famous KMP algorithm \cite{KMP} finds this $j$ in $O(|A|)$ time. \qed
\end{proof}

As a comparison, there exists a less efficient algorithm for solving Problem~1,
  which is based on a standard sorting algorithm associated with a na\"{i}ve gadget for comparing $R(A)$ and $R(B)$ -- according to Lemma~\ref{lemma:R(A)R(B)-AB-BA}, comparing $R(A)$ and $R(B)$ reduces to comparing $AB$ and $BA$, which takes $O(|A|+|B|)$ time.
  The time complexity of this alternative algorithm is higher.
  For example, when $A_1$=``aaaaaa1'', $A_2$=``aaaaaa2'', etc,
   the running time would be $\Omega(n \log n |A_1|)= \Omega(L\log n)$.

\section{The string rearrangement inequalities}

We call equation (\ref{eqn:string-rearrangement-inequalities}) right below \emph{the String Rearrangement Inequalities}.

\begin{lemma}\label{lemma:string-rearrangement-inequalities}
For non-empty words $A_1, \ldots, A_n$, where $R(A_1)\leq \ldots \leq R(A_n)$,
we claim that
  \begin{equation}\label{eqn:string-rearrangement-inequalities}
    A_1A_2\ldots A_n \leq A_{\pi_1}A_{\pi_2}\ldots A_{\pi_n} \leq A_nA_{n-1}\ldots A_1,
  \end{equation}
  for any permutation $\pi_1,\ldots,\pi_n$ of $\{1,\ldots,n\}$.

In other words, if several words are to be rearranged and concatenated into a string $S$,
     the lexicographical smallest outcome of $S$ occurs when the words are arranged so that their repeating strings are increasing,
and the lexicographical largest outcome of $S$ occurs when the words are arranged so that their repeating strings are decreasing.
Here, the repeating string of a word $A$ refers to $R(A)$.
\end{lemma}

\begin{example}\label{example:bad}
Suppose there are four given words: ``123'', ``12'', ``121'', ``1212''.
Notice that $R(121)<R(12)=R(1212)<R(123)$. Applying Lemma~\ref{lemma:string-rearrangement-inequalities},
the lexicographical smallest outcome would be ``121121212123'',
and the lexicographical largest outcome would be ``123121212121''.
The reader can verify this result easily.
\end{example}

\begin{remark}
If we sort the given words using the lexicographic order instead, the outcome of the concatenation is not optimum.
For example, we have ``12'' < ``121'' < ``1212'' <``123'',
and a concatenation in this order is not the smallest outcome, and
 a concatenation in its reverse order is neither the largest outcome.
\end{remark}

\begin{proof}[of Lemma~\ref{lemma:string-rearrangement-inequalities}]
Consider any concatenation $A_{\pi_1}\ldots A_{\pi_n}$.
If $A_1$ is not at the leftmost position,
  we swap it with its left neighbor $A_x$.
  Note that $R(A_1)\leq R(A_x)$ by assumption.
  According to Lemma~\ref{lemma:R(A)R(B)-AB-BA}, $A_1A_x\leq A_xA_1$.
  This means that the entire string becomes smaller or remains unchanged after the swapping.
  Applying several such swappings, $A_1$ will be on the leftmost position.
  Then, we swap $A_2$ to the second place. So on and so forth.
  It follows that $A_1\ldots A_n \leq A_{\pi_1}\ldots A_{\pi_n}$.

The other inequality in (\ref{eqn:string-rearrangement-inequalities}) can be proved symmetrically; proof omitted.
\qed
\end{proof}

Combining Theorem~\ref{thm:sorting-in-linear} with Lemma~\ref{lemma:string-rearrangement-inequalities}, we obtain
\begin{corollary}
Given $n$ words $A_1,\ldots,A_n$ that are to be rearranged and concatenated,
     the smallest and largest concatenation can be found in $O(\sum_i|A_i|)$ time.
\end{corollary}

Another corollary of Lemma~\ref{lemma:string-rearrangement-inequalities} is the uniqueness of the best concatenation:

\begin{corollary}
Given primitive and distinct words $A_1,\ldots,A_n$ that are to be rearranged and concatenated,
     the smallest (largest, resp.) concatenation is unique.
\end{corollary}

\begin{proof}
It follows from Lemma~\ref{lemma:string-rearrangement-inequalities} and the fact that $R(A_1),\ldots,R(A_n)$ are distinct (see Proposition~\ref{prop:dinstinct-R} below).\qed
\end{proof}

\begin{proposition}\label{prop:dinstinct-R}
For distinct primitive words $A$ and $B$, we have $R(A)\neq R(B)$.
\end{proposition}

\begin{proof}
Recall that when $A$ and $B$ are primitive and $R(A)=R(B)$, we can infer that $A=B$
(as proved in the second paragraph of the proof of Lemma~\ref{lemma:relation-M}).
Therefore, if $A$ and $B$ are primitive and distinct, $R(A)\neq R(B)$. \qed
\end{proof}

\section{A total order $\Oinf$ on words}\label{sect:order}

\begin{definition}
Given primitive words $A$ and $B$, we state that $A \Oinf B$ if $R(A)\leq R(B)$.
Notice that $\Oinf$ is \textbf{a total order on primitive words} by Proposition~\ref{prop:dinstinct-R}.
Furthermore, we extend $\Oinf$ to the scope of finite nonempty words as follows.

For non-empty words $A=S^k$ and $B=T^l$, where $S,T$ are primitive, we state that $A\Oinf B$ if
\begin{equation}
\hbox{($S=T$ and $|S|\leq |T|$), or ($S\neq T$ and $S\Oinf T$).}
\end{equation}
The symbol $\Oinf$ in the equation stands for the relation between primitive words.
\end{definition}

For example, $121 \Oinf 12 \Oinf 1212 \Oinf 121212 \Oinf 122$.

Obviously, the relation $\Oinf$ is \textbf{a total order on finite nonempty words}.

\medskip The next lemma shows that within the class of Lyndon words, the order $\Oinf$ is actually the same as the lexicographical order $\leq_{\mathsf{lex}}$ (denoted by $\leq$ for short).
(Note that Lyndon words are primitive, so the unextended $\Oinf$ is enough here.)

\begin{lemma}
Given Lyndon words $A$ and $B$ such that $A\leq B$, we have $A\Oinf B$.
\end{lemma}

\begin{proof}
Assume that $A\neq B$; otherwise we have $R(A)=R(B)$ and so $A\Oinf B$.

By the assumption $A\leq B$, we know $A<B$. Consider two cases:

\smallskip \noindent \emph{1. $|A|\geq |B|$, or $|A|<|B|$ and $A$ is not a prefix of $B$}

Combining the assumption $A<B$ with the condition of this case, we can see that the relation between $AB,BA$ is the same as that between $A,B$: In comparing $AB$ and $BA$, the result is settled before the $\min\{|A|,|B|\}$-th character.

\smallskip \noindent \emph{2. $|A|<|B|$ and $A$ is a prefix of $B$, i.e., $A$ is a proper prefix of $B$}

Assume that $B=AC$ where $C$ is nonempty. Because $B$ is a Lyndon word by assumption, $AC<CA$.
Therefore, $AB=AAC<ACA=BA$.

\smallskip In both cases, we obtain $AB<BA$. It further implies that $R(A)<R(B)$ by Lemma~\ref{lemma:R(A)R(B)-AB-BA}.
This means $A\Oinf B$. \qed
\end{proof}

In fact, it is possible to further extend $\Oinf$ to all (finite and infinite) words.
Define the repeating string of an infinite word $A$, denoted by $R(A)$, to be $A$ itself.
We state that $A\Oinf B$ if $R(A)<R(B)$ or $R(A)=R(B)$ and $|A|\leq |B|$.

\section{Conclusions}

In this paper, we present a simple proof of the ``string rearrangement inequalities'' (\ref{eqn:string-rearrangement-inequalities}).
   These inequalities have not been reported in literature to the best of our knowledge.
  We also study the algorithmic aspect of these two inequalities,
    and present a linear time algorithm for rearranging the strings so that $R(A_1)\leq \ldots R(A_n)$.
      This algorithm beats the trivial sorting algorithm by a factor of $\log n$.

   The algorithm itself is direct (indeed, it looks somewhat brute-force) and easy to implement,
      yet the analysis of its correctness and complexity is build upon nontrivial observations,
         namely, Lemma~\ref{lemma:R(A)R(B)-AB-BA}, Lemma~\ref{lemma:relation-M}, and Lemma~\ref{lemma:sum-to-L}.

In the future, it is a problem worth attacking that whether we can improve
  the running time for sorting $R(A_1),\ldots,R(A_n)$ from $O(L)$ to $O(N)$, where
    $N$ denotes the number of nodes in the trie of $A_1,\ldots,A_n$.

The order $\Oinf$ on primitive words has nice connections with repeating decimals as shown in the proof
 of Lemma~\ref{lemma:R(A)R(B)-AB-BA}. It would be interesting to know whether these connections have more applications in the study of primitive words.


\bibliographystyle{splncs04}
\bibliography{RISC}

\begin{thebibliography}{10}
\providecommand{\url}[1]{\texttt{#1}}
\providecommand{\urlprefix}{URL }
\providecommand{\doi}[1]{https://doi.org/#1}

\bibitem{order-5}
Orderings - oeiswiki (April 2022), \url{https://oeis.org/wiki/Orderings}

\bibitem{order-4}
Wikipedia: Shortlex order (April 2022),
  \url{https://en.wikipedia.org/wiki/Shortlex\_order}

\bibitem{order-3}
Alatabbi, A., Daykin, J., Rahman, M., Smyth, W.: Simple linear comparison of
  strings in v-order. In: Pal, S., Sadakane, K. (eds.) WALCOM 2014, LNCS 8344.
  pp. 80--89 (2014)

\bibitem{LW-3-debrujin-primitive}
Au, Y.: Generalized de bruijn words for primitive words and powers. Discrete
  Mathematics  \textbf{338}(12),  2320--2331 (2015).
  \doi{10.1016/j.disc.2015.05.025}

\bibitem{Run-2}
Bannai, H., I, T., Inenaga, S., Nakashima, Y., Takeda, M., Tsuruta, K.: The
  ``runs'' theorem. SIAM Journal on Computing  \textbf{46}(5),  1501--1514
  (2017). \doi{10.1137/15M1011032}

\bibitem{primitive-contextfree-97}
Berstel, J., Boasson, L.: The set of lyndon words is not context-free. Bull.
  EATCS  \textbf{63} (1997)

\bibitem{LW-2}
Chen, K., Fox, R., Lyndon, R.: Free differential calculus, iv. the quotient
  groups of the lower central series. Annals of Mathematics pp. 81--95 (1958)

\bibitem{Run-3}
Crochemore, M., Russo, L.: Cartesian and lyndon trees. Theoretical Computer
  Science  \textbf{806}, ~1--9 (2020). \doi{10.1016/j.tcs.2018.08.011}

\bibitem{order-2}
Daykin, D., Daykin, J., Smyth, W.: String comparison and lyndon-like
  factorization using v-order in linear time. In: Giancarlo, R., Manzini, G.
  (eds.) CPM 2011, LNCS 6661. pp. 65--76 (2011)

\bibitem{KMP}
D.E.~Knuth, J.M., Pratt, V.: Fast pattern matching in strings. SIAM Journal on
  Computing  \textbf{6}(2),  323--350 (1977). \doi{10.1137/0206024}

\bibitem{order-1}
Dolce, F., Restivo, A., Reutenauer, C.: On generalized lyndon words. ArXiv
  \textbf{abs/1812.04515} (2019)

\bibitem{primitive-contextfree-91}
D\"{o}m\"{o}si, P., Horv\'{a}th, S., Ito., M.: On the connection between formal
  languages and primitive words. In: Proc. First Session on Scientific
  Communication. pp. 59--67. Univ. of Oradea, Oradea, Romania (June 1991)

\bibitem{primitive-contextfree-14}
D\"{o}m\"{o}si, P., Ito, M.: Context-free languages and primitive words
  (November 2014). \doi{10.1142/7265}

\bibitem{LW-1}
Duval, J.: Factorizing words over an ordered alphabet. Journal of Algorithms
  \textbf{4}(4),  363--381 (1983). \doi{10.1016/0196-6774(83)90017-2}

\bibitem{LW-4-debrujin-original}
Fredricksen, H., Maiorana, J.: Necklaces of beads in $k$ colors and $k$-ary de
  bruijn sequences. Discrete Math.  \textbf{23},  207--210 (1979)

\bibitem{primitive-roots}
Lischke, G.: Primitive words and roots of words. Acta Univ. Sapientiae,
  Informatica  \textbf{3}(1),  5--34 (2011)

\bibitem{book-bible}
Lothaire, M.: Combinatorics on Words. Encyclopedia of Mathematics, Vol. 17,
  Addison-Wesley, MA (1983)

\bibitem{equation-1}
Lyndon, R., Sch\"{u}tzenberger, M.: The equation $a^m=b^nc^p$ in a free group.
  Michigan Math. J.  \textbf{9}(4),  289--298 (December 1962).
  \doi{10.1307/mmj/1028998766}

\bibitem{primitive-contextfree-96}
Petersen, H.: On the language of primitive words. Theoretical Computer Science
  \textbf{161},  141--156 (1996)

\bibitem{Run-1}
Smyth, W.: Computing regularities in strings: A survey. European Journal of
  Combinatorics  \textbf{34}(1),  3--14 (2013). \doi{10.1016/j.ejc.2012.07.010}

\bibitem{pattern}
Zhang, D., Jin, K.: Fast algorithms for computing the statistics of pattern
  matching. IEEE Access  \textbf{9},  114965--114976 (2021).
  \doi{10.1109/ACCESS.2021.3105607}

\end{thebibliography}

\newpage
\appendix
\section{An alternative proof of Lemma~\ref{lemma:R(A)R(B)-AB-BA}}

Below we show an alternative proof of Lemma~\ref{lemma:R(A)R(B)-AB-BA}.
This proof is less clever and much more involved (compared to the other proof in section~\ref{sect:pre}),
  yet it reflects more insights which helped us in designing our linear time algorithm.

Below we always assume that $A$, $B$, $X$, $Y$  are words.
\begin{definition}\label{def:truly_less}
\emph{Word $A$ is \textbf{truly less} than word $B$, if there exists a prefix pair $A_1A_2...A_i$ and $B_1B_2...B_i$, in which $A_1A_2...A_{i-1}$ and $B_1B_2...B_{i-1}$ are equal and $A_i$ is less than $B_i$. For convenience, let $A <_T B$ denote this case for the rest of this paper. Note that $i$ can be 1 such that $A_1$ is less than $B_1$.}
\end{definition}
\indent For any pair of nonempty words, $A$ and $B$, we can generalize 3 following properties with Definition~\ref{def:truly_less}. Note that any $X$ or $Y$ in the following properties can be any word, including empty word.\\

\begin{claim}[1]
Proposition $A <_T B$ is equivalent to $AX < BY$, if A is not prefix of B and B is not prefix of A.
\end{claim}
\begin{proof}
 If A is not prefix of B and B is not prefix of A, the proposition $AX < BY$ implies that A and B fits the case described in Definition~\ref{def:truly_less} and thus $A <_T B$ holds. The proposition $A <_T B$, by Definition~\ref{def:truly_less}, also indicates that $AX < BY$. \qed
\end{proof}

\begin{claim}[2]
If $A <_T B$, it holds that $AX <_T BY$.
\end{claim}
\begin{proof}
If $A <_T B$, by Definition~\ref{def:truly_less}, there exists a prefix pair $A_1,A_2...A_i$ and $B_1B_2...B_i$, in which $A_1A_2...A_{i-1}$ and $B_1B_2...B_{i-1}$ are equal and $A_i$ is less than $B_i$. Since $A$ is the prefix of $AX$ and $B$ is the prefix of $BY$, $AX$ and $BY$ also have the prefix pair $A_1A_2...A_i$ and $B_1B_2...B_i$ mentioned above, thus it holds that $AX <_T BY$ by Definition~\ref{def:truly_less}. \qed
\end{proof}

\begin{claim}[3]
If $A < B$ and $|A| = |B|$, it holds that $A <_T B$.
\end{claim}

\begin{proof}
 If $A < B$ and $|A| = |B|$, we can find a substring pair $A_1A_2...A_i$ and $B_1B_2...B_i$, in which $A_1A_2...A_{i-1}$ and $B_1B_2...B_{i-1}$ are equal and $A_i$ is less than $B_i$. This is exactly the case of Definition~\ref{def:truly_less}, so naturally $A <_T B$. \qed
\end{proof}

Now, we are ready for proving Lemma~\ref{lemma:R(A)R(B)-AB-BA}.

Recall that this lemma states for non-empty words $A$ and $B$, the relation between $R(A)$ and $R(B)$ is the same as the relation between $AB$ and $BA$.

We will prove $R(A)= R(B) \Leftrightarrow AB= BA $ and $R(A)< R(B) \Leftrightarrow AB< BA $.
  Note that $R(A)> R(B) \Leftrightarrow AB> BA $ can be obtained similarly.

\begin{proposition}\label{prop:1}
For nonempty words $A$ and $B$, $AB=BA\Leftrightarrow R(A)$=$R(B)$.
\end{proposition}

\begin{proof}
From $AB=BA$ or $R(A)=R(B)$, we obtain from Lemma~\ref{lemma:key-power-lemma} that $A=C^k$ and $B=C^l$ for some $C,k,l$,
  which implies that $R(A)=R(B)$ and $AB=BA$.\qed
\end{proof}

\begin{proposition}\label{prop:2}
 For nonempty words $A$ and $B$, $AB<BA\Leftrightarrow R(A)<R(B)$.
\end{proposition}

We prove the two directions separately in the following.

\begin{proof}[of $AB < BA \Rightarrow R(A) < R(B)$]

We discuss two subcases.

\medskip \noindent \emph{Subcase~1.} \emph {$|A|\leq|B|$}

\indent Let $B=A^mS$, in which $m=\deg_B(A)$ by Definition~\ref{def:deg}.

Note that target $R(A)<R(B)$ equals to $R(A)<R(A^mS)$, which then equals to $R(A)< SR(A^mS)$, by eliminating the leading $A^m$.

With $AB<BA$, we will get the relation that $AB<BA$ leads to $A^{m+1}S<A^mSA$. And $A^{m+1}S<A^mSA$ leads to $AS<SA$, which eventually leads to $AS<_TSA$.

We will prove that $AA<_TSA$. And since $AA$ is a prefix of $R(A)$ and $SA$ is a prefix of $SR(A^mS)$, proposition $R(A) < SR(A^mS)$ follows by Claim~2, proving the target proposition.

Now we prove $AA<_TSA$ in two cases.

1. If $|A|\leq|S|$, or $|A|>|S|$ but $S$ is not a prefix of $A$, note that $A$ is not a prefix of $S$, since $AS<SA$, proposition $A<_TS$ follows by Claim~1. Then $AA<_TSA$ follows by Claim~2.

2. If $|A|>|S|$ and $S$ is a prefix of $A$, let $A=ST$. Since $AS<_TSA$, we have $STS<_TSST$, then $AA = STST <_T  SST = SA$ follows by Claim~2. Thus $AA<_TSA$.\\

\medskip \noindent \emph{Subcase 2.} \emph {$|A|>|B|$}

Let $A=B^mS$, in which $m=\deg_A(B)$.

Note that target $R(A)<R(B)$ equals to $R(B^mS)<R(B)$, which then equals to $SR(B^mS)<R(B)$, by eliminating the leading $B^m$.

With AB<BA, we will get the relation that $AB<BA$ leads to $B^mSB<B^{m+1}S$. And $B^mSB<B^{m+1}S$ leads to $SB<BS$, which eventually leads to $SB<_TBS$.

We will prove that $SB<_TBB$. And since $BB$ is a prefix of $R(B)$ and $SB$ is a prefix of $SR(B^mS)$, proposition $SR(B^mS) < R(B)$ follows by Claim~2, proving the target proposition.

Now we prove $SB<_TBB$ in 2 cases.

1. If $|B|\leq|S|$ , or $|B|>|S|$ but $S$ is not a prefix of $B$, note that $B$ is not a prefix of $S$, since $SB<BS$, proposition $S<_TB$ follows by Claim~1. Then $SB<_TBB$ follows by Claim~2.

2. If $|B|>|S|$ and $S$ is a prefix of $B$, let $B=ST$. Since $SB<_TBS$, we have $SST<_TSTS$, then $SB = SST <_T  STST = BB$ follows by Claim~2. Thus $SB<_TBB$.

With both cases proved, we have $AB < BA \Rightarrow R(A) < R(B)$. \qed

\end{proof}

\newpage
\begin{proof}[of $R(A) < R(B) \Rightarrow AB<BA $]

In the following, we discuss two subcases.

\medskip \noindent \emph{Subcase~1.} \emph {$|A|\leq|B|$}.

\indent Let $B=A^mS$, in which $m=\deg_B(A)$.

Note that $AB<BA$ equals to $A^{m+1}S<A^mSA$, which equals to $AS<SA$ by eliminating the leading $A^m$.

With $R(A) < R(B)$, we will get the relation that $R(A) < R(B)$ equals to $R(A)<R(A^mS)$. And $R(A)<R(A^mS)$ equals to $R(A)<SR(A^mS)$ by eliminating the leading $A^m$.

Now we will prove $AS < SA$.

1. If $|A|<=|S|$, or $|A|>|S|$ and $S$ is not a prefix of $A$, note that $A$ is not a prefix of $S$,
since $R(A)<SR(A^mS)$, $A<_TS$ follows by Claim~1. Then $AS< SA$ follows by Claim~2.

2. If $|A|>|S|$ and $S$ is a prefix of $A$, let $A=ST$. Since $R(A)<SR(A^mS)$, we have $R(ST)<SR((ST)^mS)$, we pay attention to the prefixes with length 2*|S|+|T| of these two infinite words: $STS$ and $SST$. We argue that $STS\neq SST$ otherwise $ST=TS$, then $S,T,B,A$ are powers of a common element by Lemma~\ref{lemma:key-power-lemma}, then $R(A)=R(B)$, which is contradictory. Thus, since $R(ST)<SR((ST)^mS)$, we will have $STS<SST$. It holds that $AS=STS<SST=SA$. Thus, we end up with $AS<SA$.\\

\medskip \noindent \emph{Subcase 2.} \emph {$|A|>|B|$}.\\

Let $A=B^mS$, in which $m=\deg_A(B)$.

Note that $AB<BA$ equals to $B^mSB<B^{m+1}S$, which equals to $SB<BS$ by eliminating the leading $B^m$.

With $R(A) < R(B)$,  we will get the relation that $R(A) < R(B)$ equals to $R(B^mS)<R(B)$. And $R(B^mS)<R(B)$ equals to $SR(B^mS)<R(B)$ by eliminating the leading $B^m$.

Now we will prove $SB< BS$, in two cases.

1. If $|B|<=|S|$, or $|B|>|S|$ and $S$ is not a prefix of $B$, note that $B$ is not a prefix of $S$,
since $SR(B^mS)<R(B)$, $S<_TB$ follows by Claim~1. Then $SB< BS$ follows by Claim~2.

2. If $|B|>|S|$ and $S$ is a prefix of $B$, let $B=ST$. Since $SR(B^mS)<R(B)$, we have $SR((ST)^mS)<R(ST)$. we pay attention to the prefixes with length $2*|S|+|T|$ of these two infinite words: $SST$ and $STS$. We can argue that $SST\neq STS$ otherwise $ST=TS$, then $S,T,B,A$ are powers of a common element by Lemma~\ref{lemma:key-power-lemma}, then $R(A)=R(B)$, which is contradictory. Thus, since $SR((ST)^mS)<R(ST)$, we will have $SST<STS$. It holds that $SB=SST<STS=BS$. Thus, we end up with $SB<BS$.

With both cases proved, we have $R(A) < R(B) \Rightarrow AB < BA$.\qed
\end{proof}

Now, with both subcases proved, we have $R(A) < R(B) \Leftrightarrow AB < BA$.

\end{document}